\documentclass[12pt]{article}
\usepackage[pdftex]{graphicx}
\usepackage[T2A]{fontenc}
\usepackage{amsthm,amsmath,amssymb}

\allowdisplaybreaks

\usepackage[usenames,dvipsnames]{color}
\usepackage[pdftex,colorlinks=true,unicode,bookmarksnumbered=true]{hyperref}

\usepackage[dvipsnames]{xcolor}

\hypersetup{
  colorlinks,
  citecolor=blue!80!black,
  linkcolor=red!80!black,
  urlcolor=blue!80!black}

\usepackage[onehalfspacing]{setspace}
\usepackage{textcomp}
\usepackage{indentfirst}
\usepackage{natbib}
\usepackage{listings}
\usepackage{pdflscape}
\usepackage{enumerate}
\usepackage{hyperref}
\usepackage{verbatim}
\usepackage{bbm}
\usepackage{tikz}
\usepackage{pgfplots}
\usetikzlibrary{arrows,shapes,trees}
\usepackage{float}
 \usepackage{multirow}
 
 \usepackage[toc,page]{appendix}

\usepackage{array,longtable}

\usepackage{mathtools}

\numberwithin{equation}{section}
\newtheorem{theorem}{Theorem}

\newtheorem{assumption}{Assumption}

\newtheorem{lemma}{Lemma}

\theoremstyle{definition}
\newtheorem{remark}{Remark}
\newtheorem{example}{Example}

\renewenvironment{abstract}
 {\small
  \begin{center}
  \bfseries \abstractname\vspace{-.5em}\vspace{0pt}
  \end{center}
  \list{}{%
    \setlength{\leftmargin}{5mm}
    \setlength{\rightmargin}{\leftmargin}%
    \listparindent 1.5em
    \itemindent    \listparindent
  }%
  \item\relax}
 {\endlist}
 
 \DeclareMathOperator*{\argmin}{arg\,min}

\newcommand{\lsim}{\mathrel{\hbox{\rlap{\lower.75ex \hbox{$\sim$}} \kern-.3em \raise.4ex \hbox{$<$}}}}

\usepackage{natbib}
\usepackage{geometry}
\begin{document}

\title{LM-BIC Model Selection in Semiparametric Models\thanks{I am grateful to my advisors Frank Wolak, Han Hong, and Peter Reiss for their support and guidance, as well as to Brad Larsen and Paulo Somaini for their comments. I gratefully acknowledge the financial support from the Stanford Graduate Fellowship Fund as a Koret Fellow and from the Stanford Institute for Economic Policy Research as a B.F. Haley and E.S. Shaw Fellow.}}
\author{Ivan Korolev\thanks{Department of Economics, Binghamton University. E-mail: \mbox{\href{mailto:ikorolev@binghamton.edu}{ikorolev@binghamton.edu}}. Website: \url{https://sites.google.com/view/ivan-korolev/home}}}
\date{November 26, 2018}

\maketitle

\begin{abstract}
This paper studies model selection in semiparametric econometric models. It develops a consistent series-based model selection procedure based on a Bayesian Information Criterion (BIC) type criterion to select between several classes of models. The procedure selects a model by minimizing the semiparametric Lagrange Multiplier (LM) type test statistic from \citet{korolev_2018} but additionally rewards simpler models. The paper also develops consistent upward testing (UT) and downward testing (DT) procedures based on the semiparametric LM type specification test. The proposed semiparametric LM-BIC and UT procedures demonstrate good performance in simulations. To illustrate the use of these semiparametric model selection procedures, I apply them to the parametric and semiparametric gasoline demand specifications from \citet{yatchew_no_2001}. The LM-BIC procedure selects the semiparametric specification that is nonparametric in age but parametric in all other variables, which is in line with the conclusions in \citet{yatchew_no_2001}. The results of the UT and DT procedures heavily depend on the choice of tuning parameters and assumptions about the model errors.
\end{abstract}


\clearpage

\doublespacing

\section{Introduction}\label{introduction}

This paper considers model selection in semi- and nonparametric econometric models. It develops a consistent series-based model selection procedure for selecting between different classes of models, such as parametric, semiparametric, and nonparametric. The proposed procedure combines the features of the semiparametric LM (Lagrange Multiplier) type test developed in \citet{korolev_2018} and of the conventional parametric model selection criteria such as the BIC (Bayesian Information Criterion). It selects a model by minimizing the LM type test statistic but additionally rewards simpler models. Unlike the parametric model selection procedures, which treat the number of parameters in the model as fixed, the proposed LM-BIC procedure allows the number of parameters to grow with the sample size, thus incorporating models with nonparametric components estimated by series methods. Moreover, it takes into account that usually series methods can only approximate the true model and are misspecified for any finite number of series terms.

Model selection plays an important role in both theoretical and applied econometrics. Many papers in econometric theory have studied model selection in various settings, while applied econometricians almost always perform model selection, either formally or informally, as the true model specification is rarely, if ever, known.

Most existing model selection methods, such as the BIC or HQIC (Hannan-Quinn Information Criterion), focus on model selection in parametric models. While choosing appropriate variables in a parametric regression is an important task, it is only part of a bigger problem. Applied economists are often uncertain not only about which variables to include in a regression, but also about which class of models to use.

While parametric models are simple to estimate, easy to interpret, and yield efficient estimates when correctly specified, they may be too restrictive, fail to capture heterogeneity in the data well, and lead to inconsistent estimates under misspecification. On the other hand, nonparametric models are extremely flexible, but they may violate restrictions imposed by economic theory and suffer from the curse of dimensionality. 

Semiparametric models offer an attractive middle ground by combining the flexibility of nonparametric models with parametric assumptions about certain components of the model, which alleviates the curse of dimensionality and helps achieve regular models. Nevertheless, semiparametric models yield inconsistent estimates when misspecified and inefficient estimates when the true model is parametric.

Thus, specification choice poses a big challenge for empirical work in economics, and it may be desirable to have econometric methods that help find the most appropriate model. One way of achieving this goal is specification tests, which are reasonable when the researcher has a preferred model (say, obtained from economic theory) and wants to test whether it provides a statistically adequate description of the data. But if the researcher does not have such a model and needs to choose from a wide range of models, she may want to use a model selection procedure that compares all candidate models simultaneously.

One such procedure that has long been known in econometrics is the BIC, which performs model selection in parametric models by penalizing the number of parameters to be estimated, as adding more parameters always improves fit. The BIC was first introduced by \citet{schwarz_1978} as a Bayesian solution to model selection in parametric models. \citet{hannan_quinn_1979} and \citet{hannan_1980} prove that the BIC, as well as its modification called HQIC, consistently selects the true model under fairly general assumptions. \citet{shao_1997} develops a unified asymptotic theory for various model selection procedures, including the BIC.

\citet{andrews_1999} shows that the BIC can be used for consistent moment selection in GMM (Generalized Method of Moments) models, while \citet{andrews_lu_2001} and \citet{hong_et_al_2003} extend the BIC to both model and moment selection in GMM and EL (Empirical Likelihood) estimation and prove its consistency.

Another procedure for model selection is based on the Lasso and penalizes the sum of absolute values of the parameters (i.e. the $\ell_1$ norm of the parameters vector). The Lasso was first proposed by \citet{tibshirani_1996}. The conditions for its consistency in the context of parametric model selection are developed in \citet{zhao_yu_2006} and \citet{zou_2006}. \citet{caner_2009} extends the Lasso to model selection in GMM estimation and proves its consistency, while \citet{liao_2013} shows that the Lasso can also be used for consistent moment selection in GMM.

As for model selection in semi- and nonparametric models, \citet{huang_et_al_2010} prove that the group Lasso can consistently select the correct regressors in nonparametric additive models. \citet{lian_2014} further shows that a BIC type criterion is also consistent in this context.

\citet{belloni_chernozhukov_2013} study the performance of the Lasso and post-Lasso OLS in high-dimensional sparse models. They analyze the selection properties of the Lasso and develop an asymptotic theory for resulting post-selection estimators.

The focus of most of these papers has been on parametric models or on additive semiparametric models, while I investigate the use of BIC-type criteria in general semiparametric contexts. For example, suppose that the researcher wants to estimate the average treatment effect of treatment $D_i$ on outcome $Y_i$ while controlling for observables $X_i$ and $Z_i$, assuming that the treatment assignment is as good as random after accounting for observable differences. A fairly general way of doing this is to estimate the nonparametric (in $X_i$ and $Z_i$) specification:
\[
Y_i = \beta_1 D_i + g(X_i, Z_i) + \varepsilon_i, \quad E[\varepsilon_i | X_i, Z_i, D_i] = 0
\]

However, in practice this specification may be infeasible because of the curse of dimensionality, so the researcher may want to determine if more restricted models are appropriate. One such model which is commonly used in empirical work is fully parametric, with $g(X_i, Z_i) = \beta_0 + X_i' \gamma_1 + Z_i' \gamma_2$. Another possible specification is semiparametric , e.g. nonparametric in $X_i$ but parametric in $Z_i$, wtih $g(X_i, Z_i) = g_1(X_i) + Z_i' \gamma_2$. Yet another possible model is additively separable in $X_i$ and $Z_i$, with $g(X_i, Z_i) = g_1(X_i) + g_2(Z_i)$.

In this paper, I develop a model selection technique that uses series methods to estimate semi- and nonparametric models. It performs model selection by comparing the values of the model selection criterion that is based on the semiparametric LM test statistic and additionally penalizes more complicated models. More specifically, the model selection procedure minimizes the following model selection criterion:
\begin{align}\label{msc_criterion}
MSC(s) = \frac{\xi_s - r_{n,s} \kappa_n}{\sqrt{2 r_{n,s}}} = \frac{\xi_s - r_{n,s}}{\sqrt{2 r_{n,s}}} - \frac{r_{n,s} (\kappa_n - 1)}{\sqrt{2 r_{n,s}}} ,
\end{align}
where $\xi$ is the quadratic form in the semiparametric model residuals, $s$ denotes a particular series form, $r_{n,s}$ is the number of restrictions the series form $s$ imposes on the fully nonparametric model, and $\kappa_n \to \infty$ at a certain rate as $n \to \infty$.\footnote{I will formally introduce the notation in the next section.} A natural choice of $\kappa_n$ would be $\kappa_n = \ln{n}$, but other functions satisfying certain rate conditions will also work. Intuitively, $\frac{\xi_s - r_{n,s}}{\sqrt{2 r_{n,s}}}$ is the semiparametric LM type test statistic as in \citet{korolev_2018}, while $\frac{r_{n,s} (\kappa_n - 1)}{\sqrt{2 r_{n,s}}}$ is the penalty term that rewards simpler models (i.e. models with large $r_{n,s}$). For the fully unrestricted nonparametric model, $r_{n,s} = 0$, and in that case $MSC(s) = 0$. Model search takes place not over individual coefficients but over the groups that correspond to different classes of models.

The remainder of the paper is organized as follows. Section~\ref{setup} introduces the model, describes the model selection problem, and introduces the model selection criterion. Section~\ref{asymptotics} develops the asymptotic theory for the proposed model selection procedure. Sections~\ref{upward_testing} and~\ref{downward_testing} discuss an upward and downward testing procedures based on the LM type specification test from \citet{korolev_2018}. Section~\ref{simulations} studies the behavior of the proposed model selection procedure in simulations. Section~\ref{empirical_example} applies the proposed method to the household gasoline demand dataset from \citet{yatchew_no_2001}. Section~\ref{conclusion} concludes.

Appendix~\ref{appendix_tables_figures} collects all tables. Appendix~\ref{appendix_proofs} contains proofs of technical results. 

\section{Setup and Definitions}\label{setup}

\subsection{Classes of Models}

Let $(Y_i, X_i')' \in \mathbb{R}^{1+d_x}$, $d_x \in \mathbb{N}$, $i = 1, ..., n$, be independent and identically distributed random variables with $E[Y_i^2] < \infty$. Then there exists a measurable function $g$ such that $g(X_i) = E[Y_i | X_i]$ a.s., and the nonparametric model can be written as
\[
Y_i = g(X_i) + \varepsilon_i, \quad E[\varepsilon_i | X_i] = 0
\]

\begin{remark}

In fact, the most flexible model under consideration does not have to be fully nonparametric. In the treatment effects example above, the most flexible model is given by
\[
Y_i = \beta_1 D_i + g(X_i, Z_i) + \varepsilon_i, \quad E[\varepsilon_i | X_i, Z_i, D_i] = 0
\]
As long as this semiparametric model is general enough to include the true data generating process, all the results obtained below will still hold, and model selection can be performed on the nonparametric part $g(X_i, Z_i)$ under the maintained assumption of additive separability in $D_i$.

\end{remark}

This nonparametric model is very flexible, but it may suffer from the curse of dimensionality or may be difficult to interpret. Because of this, the researcher may prefer more restrictive models, such as semiparametric ones. If these models are correctly specified, they are likely to provide more efficient estimates of the parameters of interest; however, if these models are misspecified, the resulting estimates will be inconsistent.

Thus, the researcher's goal is to compare different parametric, semiparametric, and nonparametric models and to select the one which is correctly specified yet parsimonious. A semiparametric model $c \in \mathcal{C}$ can be written as
\[
Y_i = f_c(X_i, \theta_c, h_c) + e_{c,i}
\]
where $f_c: \mathcal{X} \times \Theta_c \times \mathcal{H}_c \to \mathbb{R}$ is a known function, $\theta_c \in \Theta_c \subset \mathbb{R}^{d_c}$ is a finite-dimensional parameter, and $h_c \in \mathcal{H}_c = \mathcal{H}_{1,c} \times ... \times \mathcal{H}_{q_c,c}$ is a vector of unknown functions. Let $\mathcal{C}^*$ denote the set of correctly specified models. Formally,
\[
\mathcal{C}^* = \{ c \in \mathcal{C} : P_{X} \left(g(X_i) = f_c(X_i, \theta_{c,0}, h_{c,0}) \right) = 1 \text{ for some } \theta_{c,0} \in \Theta_c, h_{c,0} \in \mathcal{H}_c \}
\]

Several different models can be correctly specified simultaneously, and the researcher's goal is to find the most parsimonious of them, i.e. the one with the lowest dimensionality of the nonparametric component or with the fastest convergence rate.\footnote{I will define what precisely consistent model selection means in semi- and nonparametric models later.} Before discussing the model selection problem in more detail, I will introduce series methods that will be used in estimation and model selection.

\subsection{Series Methods}

As in \citet{korolev_2018}, I use series for the purposes of estimation and model selection. For any variable $z$, let $Q^{a_n}(z) = (q_1(z), ..., q_{a_n}(z))'$ be an $a_n$-dimensional vector of approximating functions of $z$, where the number of series terms $a_n$ is allowed to grow with the sample size $n$. Possible choices of series functions include:

\begin{enumerate}[(a)]

\item

Power series. For univariate $z$, they are given by:
\begin{align}\label{series}
Q^{a_n}(z) = (1, z, ..., z^{a_n-1})'
\end{align}

\item

Splines. Let $s$ be a positive scalar giving the order of the spline, and let $t_{1}, ..., t_{a_n - s - 1}$ denote knots.  Then for univariate $z$, splines are given by:
\begin{align}\label{splines}
Q^{a_n}(z) = (1, z, ..., z^s, \mathbbm{1}\{ z > t_1 \} (z-t_1)^s, ..., \mathbbm{1}\{ z > t_{a_n - s -1} \} (z - t_{a_n - s -1})^s)
\end{align}

\end{enumerate}

Multivariate power series or splines can be formed from products of univariate ones.

\subsection{Model Selection Problem}

Because the class of models $\mathcal{C}$ that are allowed includes semiparametric models and because I use series methods to estimate different models, the formal statement of model selection problem is tricky. In particular, all models written in series forms will only provide approximations to semiparametric models and will most likely be misspecified for any fixed number of series terms. However, if a given series form can approximate a particular model that is correctly specified, this misspecification should disappear as the number of series terms grows to infinity. In this section I formally state the model selection problem and define what consistent model selection means for semiparametric models.

To fix ideas, it is useful to consider model selection in parametric models first and contrast it with model selection in semiparametric models. Usually model selection procedures seek to achieve two goals: select a model that is (a) correctly specified and (b) parsimonious. In parametric models, the model selection problem can usually be formulated as follows. Suppose that the model is given by
\[
Y_i = \sum_{j=1}^{k}{X_{ji} \beta_j} + \varepsilon_i, \quad E[\varepsilon_i | X_i] = 0,
\]
and some of the $\beta_j$s can be zero. Then the goal of model selection is to find out which $\beta_j$s are equal to zero and select the model which includes all of the relevant regressors (i.e. is correctly specified) but none of the irrelevant ones (i.e. is parsimonious).

It might be natural to think of model selection in semiparametric models as the problem of choosing $a_n$ or picking nonzero elements in the series expansion $h(z) = \sum_{j=1}^{a_n}{q_j(z) \beta_j}$, but this is not what I am doing here. The moment selection problem is semiparametric models will be formulated differently, but it will still seek to combine parsimoniousness and correct specification.

First, fix $a_n$, the number of series terms used in series expansions of different semiparametric models. This fixed $a_n$ will lead to a certain number of series terms $m_{n,s}(a_n)$ in every series form $s \in \mathcal{S}$ under consideration. The model selection procedure will treat $a_n$ as fixed and will aim to select the most parsimonious model, i.e. the model with the lowest $m_{n,s}(a_n)$.

But, of course, it is not enough to select the simplest model; the selected model should also be correctly specified. As discussed above, if the true model is semi- or nonparametric, no series form will be exactly correct for any fixed $a_n$. Thus, to define what correct specification means in semiparametric settings, I will let $a_n \to \infty$ and look for series forms that can approximate the true semiparametric model well as $a_n$ grows.

In other words, the goal is to select the series form that (a) can approximate the true model as the number of series terms $a_n$ in every series expansion grows; (b) for each fixed $a_n$, has the lowest number of terms $m_{n,s}(a_n)$ among all models that satisfy (a). Note that the nonparametric model $Y_i = g(X_i) + \varepsilon_i$, written in a series form, will always satisfy (a), but it will have the largest number of parameters $m_{n,s}(a_n)$ among all models for a given $a_n$.

I now state the model selection problem formally. Write model $c$ in series form $s$ as
\begin{align}\label{semiparametric_series_form}
Y_i = f_c(X_i,\theta_c,h_c) + e_{c,i} = W_{s,i}' \beta_{1,s} + R_{c,s,i} + e_{c,i} = W_{si}' \beta_{1,s} + \eta_{s,c,i},
\end{align}
where $W_{s,i} := W_s^{m_{n,s}}(X_i) = (W_{s,1}(X_i), ..., W_{s,m_{n,s}}(X_i))'$ are appropriate regressors or basis functions that are used in the series form $s$, $R_{c,s,i} := \left( f_c(X_i,\theta_c,h_c) - W_{s,i}' \beta_{1,s} \right)$ is the approximation error, and $\eta_{c,s,i} := e_{c,i} + R_{c,s,i}$ is the composite error term.

For a given series form $s$ and number of series terms $m_{n,s}(a_n)$, the estimate of $\beta_{1,s}$ is
\[
\tilde{\beta}_{1,s} = (W_{s}' W_{s})^{-1} W_{s}' Y,
\]
where $W_{s} = (W_{s,1}, ..., W_{s,n})'$, $Y = (Y_1, ..., Y_n)'$.

Each semiparametric model $c$ can be written in a series form $s$ defined by the series terms used $W_{s}^{m_{n,s}}(x)$ and the number of series terms $m_{n,s}(a_n)$ as a function of $a_n$. However, not all series forms $s$ provide a good approximation of the model $c$. Among all series forms $s \in \mathcal{S}$, I need to define the ones that correspond to the semiparametric model $c$. Let
\[
\mathcal{S}^*(c) = \{ s \in \mathcal{S}: \sup_{x \in \mathcal{X}}{| f(x,\theta_{c},h_{c}) - W_s^{m_{n,s}}(x)' \beta_{1,s} |} \to 0 \quad \text{as} \quad a_n \to \infty \}
\]
be all series forms $s$ that can approximate a given semiparametric model $c$. 

Next, I consider \textit{all} correctly specified models $c \in \mathcal{C}^*$ and define \textit{all} series forms that can approximate at least one correctly specified model. Namely, let
\[
\mathcal{S}_u^* = \cup_{c \in \mathcal{C}^*}{ \mathcal{S}^*(c) }
\]
be all ``correctly specified'' series forms. In other words, $\mathcal{S}_u^*$ includes all series forms $s$ that can approximate some correctly specified model $c$ as the sample size grows.

As I will show, the model selection procedures I propose will select the most parsimonious member of $\mathcal{S}_u^*$, i.e. the series form $s \in \mathcal{S}_u^*$ with the smallest $m_{n,s}(a_n)$ for every fixed $a_n$.

\begin{example}

Suppose that $X_i = (X_{1i},X_{2i})$. The researcher considers three models: a semiparametric partially linear model $Y_i = X_{1i} \beta_1 + g_{PL,2}(X_{2i}) + \varepsilon_i$, an additive semiparametric model $Y_i = g_{A,1}(X_{1i}) + g_{A,2}(X_{2i}) + \varepsilon_i$, and a nonparametric model $Y_i = g_{NP}(X_{1i},X_{2i}) + \varepsilon_i$. 

In my notation, $f_{PL}(X_i,\theta_{PL},h_{PL}) = X_{1i} \beta_{PL} + g_{PL,2}(X_{2i})$ with $\theta_{PL} = \beta_{PL}$ and $h_{PL} = g_{PL,2}$;  $f_A(X_i,\theta_A,h_A) = g_{A,1}(X_{1i}) + g_{A,2}(X_{2i})$ with $h_A = (g_{A,1}, g_{A,2})'$, $f_{NP}(X_i,\theta_{NP},h_{NP}) = g_{NP}(X_{1i},X_{2i})$ with $h_{NP} = g_{NP}$.

Let $Q_1^{a_n}(x_1)$ and $Q_2^{a_n}(x_2)$ be series terms used to approximate functions of $x_1$ and $x_2$ correspondingly. Suppose that they both include a constant term, as do power series and splines. Denote by $\tilde{Q}^{a_n-1}(\cdot)$ the series term that exclude the constant from $Q^{a_n}(\cdot)$.

Suppose that the researcher considers the following series forms: $W_{1}^{m_{n,1}}(X_i) = (X_{1i}, Q_2^{a_n}(X_{2i})')'$ with ${m_{n,1}}(a_n) = a_n + 1$; $W_{2}^{m_{n,2}}(X_i) = (Q_1^{a_n}(X_{1i})', X_{2i})'$ with ${m_{n,2}}(a_n) = a_n + 1$; $W_{3}^{m_{n,3}}(X_i) = (1, \tilde{Q}_1^{a_n-1}(X_{1i})', \tilde{Q}_2^{a_n-1}(X_{2i})')'$ with ${m_{n,3}} = 2 a_n - 1$; and $W_{4}^{m_{n,4}}(X_i) = vec(Q_1^{a_n}(X_{1i}) Q_2^{a_n}(X_{2i})')$ with ${m_{n,4}}(a_n) = a_n^2$. Then $\mathcal{S}^*(PL) = \{ 1, 3, 4 \}$, $\mathcal{S}^*(A) = \{ 3, 4 \}$, and $\mathcal{S}^*(NP) = \{ 4 \}$.

If the true model is partially linear, e.g. $Y_i = X_{1i} \beta_1 + \exp(X_{2i}) \beta_2 + \varepsilon_i$, then $\mathcal{S}_u^* = \{ 1, 3, 4 \}$ and the model selection procedure should select the first series form with $W_{1}^{m_{n,1}}(X_i) = (X_{1i}, Q_2^{a_n}(X_{2i})')'$, because it is the simplest series form that accounts for nonlinearities in $X_{2i}$. If the true model is nonlinear in $X_{1i}$ as well, e.g. $Y_i = X_{1i}^2 \beta_1 + \exp(X_{2i}) \beta_2 + \varepsilon_i$, then $\mathcal{S}_u^* = \{ 3, 4 \}$ and the model selection procedure should select the third series form with $W_{3}^{m_{n,3}}(X_i) = (1, \tilde{Q}_1^{a_n-1}(X_{1i})', \tilde{Q}_2^{a_n-1}(X_{2i})')'$, because it is the simplest series form that accounts for nonlinearities in both $X_{1i}$ and $X_{2i}$. If the true model includes interactions between $X_{1i}$ and $X_{2i}$, e.g. $Y_i = X_{1i}^2 \beta_1 + \exp(X_{2i}) \beta_2 + X_{1i} X_{2i} \beta_3 + \varepsilon_i$, then $\mathcal{S}_u^* = \{ 4 \}$ and the model selection procedure should select the fourth series form with $W_{4}^{m_{n,4}}(X_i) = vec(Q_1^{a_n}(X_{1i}) Q_2^{a_n}(X_{2i})')$, because none of the simpler models accounts for possible interactions between $X_{1i}$ and $X_{2i}$.

\end{example}

\subsection{Model Selection Criterion}

Next, I introduce additional series terms $T_{s,i} := T_s^{r_{n,s}}(X_i)$ and let $P_i := P^{k_{n,s}}(X_i) := (W_{s,i}', T_{s,i}')'$, so that $P_i$ can approximate a fully nonparametric model as $a_n \to \infty$. Note, however, that $a_n$ will be treated as fixed by the model selection procedure, so that it is possible to compare the complexity of different models $m_{n,s}(a_n)$. Here $m_{n,s}(a_n)$ is the number of series terms in the series form $s$ and $r_{n,s}(a_n)$ is the number of restrictions that the series form $s$ imposes on a fully nonparametric model for a given $a_n$. The model selection procedure is based on the LM type quadratic form, as in \citet{korolev_2018}:\footnote{For any vector $V_i$, let $V = (V_1, ..., V_n)'$ be the matrix that stacks all observations together.}
\[
\xi_s = \tilde{\varepsilon}_s' P (\tilde{\sigma}_s^2 P'P)^{-1} P' \tilde{\varepsilon}_s,
\]
where $\tilde{\varepsilon}_s = Y - W_s \tilde{\beta}_{1,s}$ are the semiparametric residuals for the series form $s$ and $\tilde{\sigma}_s^2 = \tilde{\varepsilon}_s' \tilde{\varepsilon}_s/n$. A heteroskedasticity robust version of the LM-BIC procedure is based on 
\[
\xi_{HC,s} = \tilde{\varepsilon}_s' \tilde{T}_s (\tilde{T}_s' \tilde{\Sigma}_s \tilde{T}_s)^{-1} \tilde{T}_s' \tilde{\varepsilon}_s,
\]
where $\tilde{\Sigma}_s = diag(\tilde{\varepsilon}_{s,i}^2)$ and $\tilde{T}_s = M_{W_s} T_s = (I - W_s (W_s' W_s)^{-1} W_s) T_s$ are the residuals from the regression of each element of $T_{s,i}$ on $W_{s,i}$.

This form is analogous to the parametric LM test statistic or overidentifying restrictions test statistic, but the dimensionality of $P$, $k_n$, is allowed to grow with $n$ to accommodate semi- and nonparametric models. The model selection criterion selects the model with the lowest value of the LM test statistic but additionally rewards parsimonious models that impose more restrictions:
\[
MSC(s) = \frac{\xi_s - r_{n,s} \kappa_n}{\sqrt{2 r_{n,s}}},
\]
where $r_{n,s}$ is the number of restrictions the series form $s$ imposes on the fully nonparametric model and $\kappa_n \to \infty$ as $n \to \infty$. The heteroskedastiticy robust model selection criterion is
\[
MSC_{HC}(s) = \frac{\xi_{HC,s} - r_{n,s} \kappa_n}{\sqrt{2 r_{n,s}}}
\]

\section{Asymptotic Theory for Semiparametric LM-BIC}\label{asymptotics}

I impose the following assumptions that are similar to the ones in \citet{korolev_2018}.

\begin{assumption}\label{dgp}

$(Y_i, X_i')' \in \mathbb{R}^{1+d_x}, d_x \in \mathbb{N}, i = 1, ..., n$ are i.i.d. random draws of the random variables $(Y, X')'$, and the support of $X$, $\mathcal{X}$, is a compact subset of $\mathbb{R}^{d_x}$.

\end{assumption}

Next, I define the error term and impose two moment conditions.

\begin{assumption}\label{errors_unified}

Let $\varepsilon_i = Y_i - E[Y_i | X_i]$. The following two conditions hold:

\begin{enumerate}[(a)]

\item

$0 < \sigma^2(x) = E[\varepsilon_i^2 | X_i = x] < \infty$.

\item

$E[\varepsilon_i^4 | X_i]$ is bounded.

\end{enumerate}

\end{assumption}

The following assumption deals with the behavior of the approximating series functions. From now on, let $|| A || = [tr(A'A)]^{1/2}$ be the Euclidian norm of a matrix $A$.

\begin{assumption}\label{series_norms_eigenvalues}

For each $m_{s}$, $r_{s}$, and $k$ there are matrices $B_{1,s}$ and $B_{2,s}$ such that, for $\tilde{W}_s^{m_{s}}(x) = B_{1,s} W_s^{m_{s}}(x)$, $\tilde{T}_s^{r_{s}}(x) = B_{2,s} T_s^{r_{s}}(x)$, and $\tilde{P}_s^{k}(x) = (\tilde{W}_s^{m_{s}}(x)', \tilde{T}_s^{r_{s}}(x)')'$,

\begin{enumerate}[(a)]

\item

There exists a sequences of constants $\zeta(\cdot)$ that satisfies the conditions $\sup_{x \in \mathcal{X}} || \tilde{W}^{m_{s}}(x) || \leq \zeta(m_{s})$, $\sup_{x \in \mathcal{X}} || \tilde{T}^{r_{s}}(x) || \leq \zeta(r_{s})$, and $\sup_{x \in \mathcal{X}} || \tilde{P}^{k}(x) || \leq \zeta(k)$.

\item

The smallest eigenvalue of $E[\tilde{P}^{k}(X_i) \tilde{P}^{k}(X_i)']$ is bounded away from zero uniformly in $k$.

\end{enumerate}

\end{assumption}

The following assumption states that series methods should be able to approximate unknown functions sufficiently well as the number of series terms grows.

\begin{assumption}\label{series_approx}
If $s \in \mathcal{S}^*(c)$, there exists $\alpha > 0$ such that
\[
\sup_{x \in \mathcal{X}}{| f_c(x,\theta_c,h_c) - W^{m_{n,s}}(x)' \beta_{1,s} |} = O(m_{n,s}^{-\alpha})
\]
\end{assumption}

Then I can use the following result.

\begin{lemma}[\citet{li_racine_2007}, Theorem 15.1]\label{series_f_rates}

Suppose that $c \in \mathcal{C}^*$. Let $f_c(x) = f_c(x,\theta_{c},h_{c})$, $f_{c,i} = f_c(X_i)$, $\tilde{f}_{s}(x) = W_s^{m_{n,s}}(x)' \tilde{\beta}_{1,s}$, and $\tilde{f}_{s,i} = \tilde{f}_{s}(X_i)$. Under Assumptions \ref{dgp}, \ref{series_norms_eigenvalues}, and \ref{series_approx}, the following is true:

\[
\sup_{x \in \mathcal{X}} | \tilde{f}_s(x) - f_c(x) | = O_p \left( \zeta(m_{n,s}) \left( \sqrt{m_{n,s}/n}+m_{n,s}^{-\alpha} \right) \right)
\]

\end{lemma}

These assumptions will allow me to use the results from \citet{korolev_2018}, e.g. Theorems 1, 2, and 3. Then I obtain the following result.

\begin{theorem}\label{lm_bic_consistency}
Suppose that $\mathcal{C}^*$ is nonempty and Assumptions~\ref{dgp}--\ref{series_approx} hold. Assume that $\sigma^2(x) = \sigma^2$, $0< \sigma^2 < \infty$, for all $x \in \mathcal{X}$ and the following rate conditions hold for every series form $s \in \mathcal{S}$ except for the fully nonparametric model:

\begin{align}
\label{rate_cond_r_1_ms} \zeta(k_n)^2 k_n r_{n,s}^{1/2}/n &\to 0 \\
\label{rate_cond_r_2_ms} \zeta(r_n) r_{n,s} / n^{1/2} &\to 0 \\
\label{rate_cond_r_3_ms} \zeta(k_n) m_{n,s}^{1/2} k_n^{1/2}/n^{1/2} &\to 0 \\
\label{rate_cond_r_4_ms} n m_{n,s}^{- 2\alpha}/ r_{n,s}^{1/2} &\to 0 \\
\label{rate_cond_r_5_ms} \zeta(r_{n,s})^2/n^{1/2} &\to 0
\end{align}

Suppose that the series form $\hat{s}$ is selected by minimizing the model selection criterion $MSC(s)$ over all series forms $\mathcal{S}$ under consideration:
\[
\hat{s} = \argmin_{s \in \mathcal{S}}{MSC(s)} = \argmin_{s \in \mathcal{S}}{\frac{\xi_s - r_{n,s} \kappa_n}{\sqrt{2 r_{n,s}}}}
\]

Suppose that $\kappa_n \to \infty$ and $\kappa_n r_{n,s}/n \to 0$ as $n \to \infty$ for all $s \in \mathcal{S}$. Then the selected series form $\hat{s}$ is consistent in the following sense: 

\begin{enumerate}

\item

$\hat{s} \in \mathcal{S}_u^*$.
\item

$m_{n,\hat{s}}(a_n) \leq m_{n,s'}(a_n)$ for any $s' \in \mathcal{S}_u^*$.

\end{enumerate}

\end{theorem}

This theorem makes two statements. First, the selected series form $\hat{s}$ approximates a semiparametric model $c$ that is correctly specified. Second, the selected series form $\hat{s}$ has the lowest dimension among all series forms that can approximate correctly specified models.

A similar result can be obtained for the heteroskedasticity robust model selection procedure:
\begin{theorem}\label{lm_bic_consistency_hc}
Suppose that $\mathcal{C}^*$ is nonempty and Assumptions~\ref{dgp}--\ref{series_approx} hold. Assume that the following rate conditions hold for every series form $s \in \mathcal{S}$ except for the fully nonparametric model:

\begin{align}
\label{rate_cond_r_1_hc_ms} (m_{n,s}/n + m_{n,s}^{-2\alpha}) \zeta(r_{n,s})^2 r_{n,s}^{1/2} &\to 0 \\
\label{rate_cond_r_2_hc_ms} \zeta(r_{n,s}) r_{n,s} / n^{1/2} &\to 0 \\
\label{rate_cond_r_3_hc_ms} \zeta(k_n) m_{n,s}^{1/2} k_n^{1/2}/n^{1/2} &\to 0 \\
\label{rate_cond_r_4_hc_ms} n m_{n,s}^{-2\alpha}/ r_{n,s}^{1/2} &\to 0 \\
\label{rate_cond_r_5_hc_ms} \zeta(r_{n,s})^2/n^{1/2} &\to 0
\end{align}

Also assume that $|| \hat{\Omega}_s - \tilde{\Omega}_s || = o_p(r_{n,s}^{-1/2})$, where $\tilde{\Omega}_s = T_s' \tilde{\Sigma}_s T/n$ and $\hat{\Omega}_s = \tilde{T}_s' \tilde{\Sigma}_s \tilde{T}_s/n$.

Suppose that the series form $\hat{s}$ is selected by minimizing the model selection criterion $MSC_{HC}(s)$ over all series forms $\mathcal{S}$ under consideration:
\[
\hat{s} = \argmin_{s \in \mathcal{S}}{MSC_{HC}(s)} = \argmin_{s \in \mathcal{S}}{\frac{\xi_{HC,s} - r_{n,s} \kappa_n}{\sqrt{2 r_{n,s}}}}
\]

Suppose that $\kappa_n \to \infty$ and $\kappa_n r_{n,s}/n \to 0$ as $n \to \infty$ for all $s \in \mathcal{S}$. Then the selected series form $\hat{s}$ is consistent in the following sense: 

\begin{enumerate}

\item

$\hat{s} \in \mathcal{S}_u^*$.

\item

$m_{n,\hat{s}}(a_n) \leq m_{n,s'}(a_n)$ for any $s' \in \mathcal{S}_u^*$.

\end{enumerate}

\end{theorem}

\begin{remark}

In Theorems~\ref{lm_bic_consistency} and~\ref{lm_bic_consistency_hc}, the condition $\kappa_n \to \infty$ is not needed for consistency if $\mathcal{S}$ only contains semiparametric models with $r_{n,s}(a_n) \to \infty$ as $n \to \infty$. However, because I allow for fully parametric models with $r_{n,s}$ fixed, I require $\kappa_n$ to grow.

\end{remark}

\section{Upward Testing Procedure}\label{upward_testing}

In this and next sections I study testing procedures that sequentially apply the LM type specification test from \citet{korolev_2018} with adjusted critical values to different models. The upward testing (UT) procedure considered in this section starts from the most restrictive (e.g. fully parametric) model and carries out tests with progressively smaller $r_{n,s}$, the number of restrictions, until it finds a model that is not rejected.

Let 
\[
t_{s,m_{n,s}} = \frac{\xi_{s,m_{n,s}} - r_{n,s}}{\sqrt{2 r_{n,s}}}
\]
be the value of the LM type $t$ test statistic for model $s$ when we use $m_{n,s}(a_n)$ series terms. Let $\gamma_n > 0$ denote the critical value for the UT test. Define $\hat{m}_{n}^{UT}$ such that 
\[
\min_{s \in \mathcal{S}: m_{n,s} = \hat{m}_{n}^{UT}}{t_{s,m_{n,s}}} \leq \gamma_n \quad \text{and} \quad 
\min_{s \in \mathcal{S}: m_{n,s} < \hat{m}_{n}^{UT}}{t_{s,m_{n,s}}} > \gamma_n
\]

Let
\[
\hat{s}^{UT} = \argmin_{s \in \mathcal{S}: m_{n,s} = \hat{m}_{n}^{UT}}{t_{s,m_{n,s}} }
\]

In other words, $\hat{m}_{n}^{UT}$ is such that the null is rejected for all series forms with $m_{n,s} < \hat{m}_{n}^{UT}$ but is not rejected for at least one series form with  $m_{n,s} = \hat{m}_{n}^{UT}$. $\hat{s}^{UT}$ is such that it minimizes the LM type test statistic among all series forms $s$ with $m_{n,s} =  \hat{m}_{n}^{UT}$.

\begin{theorem}\label{upward_consistency}

Assume that all conditions of Theorem~\ref{lm_bic_consistency} hold. Assume that $\gamma_n \to \infty$ and $\gamma_n = o(n/\sqrt{r_{n,s}})$ for every model $s \in \mathcal{S}$ except for the nonparametric model. Then the selected series form $\hat{s}^{UT}$ is consistent in the following sense: 

\begin{enumerate}

\item

$\hat{s}^{UT} \in \mathcal{S}_u^*$.
\item

$\hat{m}_{n}^{UT}(a_n) \leq m_{n,s'}(a_n)$ for any $s' \in \mathcal{S}_u^*$.

\end{enumerate}

\end{theorem}

\section{Downward Testing Procedure}\label{downward_testing}

The downward testing (DT) procedure considered in this sections starts from the least restrictive (e.g. fully nonparametric) model and carries out tests with progressively larger $r_{n,s}$, the number of restrictions, until it finds a model that is rejected.

Define
\[
\hat{m}_{n}^{DT} = \min_{m \in \mathcal{M}_0^{**}}{m}, \text{ where } \mathcal{M}_0^{**} = \{ m^{**} : \min_{s \in \mathcal{S}: m_{n,s} =m^{***}}{t_{s,m_{n,s}}} \leq \gamma_n \text{ for all } m^{***} \geq m^{**} \}
\]
and
\[
\hat{s}^{DT} = \argmin_{s \in \mathcal{S}: m_{n,s} = \hat{m}_{n}^{DT}}{t_{s,m_{n,s}} }
\]

In other words, $\mathcal{M}_0^{**}$ is such that for each $m^{**} \in \mathcal{M}_0^{**}$ and each $m^{***} \geq m^{**}$ there is a series form $s$ with $m_{n,s} = m^{***}$ which is not rejected. $\hat{m}_{n}^{DT}$ is the smallest element of that set. $\hat{s}^{DT}$ is such that it minimizes the LM type test statistic among all series forms $s$ with $m_{n,s} =  \hat{m}_{n}^{DT}$.

I need to make as additional assumption to ensure that the DT procedure is consistent.

\begin{assumption}\label{no_gaps}

Let $s^{**}$ be the series form such that $s^{**} \in \mathcal{S}_u^*$ and $m_{n,s^{**}} \leq m_{n,s}$ for any $s \in \mathcal{S}_u^*$. For any $m_{n,s} > m_{n,s^{**}}$, where $s \in \mathcal{S}$, there is a series form $s' \in \mathcal{S}_u^*$ with $m_{n,s'} = m_{n,s}$.

\end{assumption}

Intuitively, because the DT procedure starts with the largest model and tries smaller and smaller $m_{n,s}$, it will only arrive at the smallest correctly specified model if there are no ``gaps'' in the series form space. In other words, there is no $m_{n,s}$ that is larger than the smallest correctly specified one for which there is no correctly specified model.

\begin{theorem}\label{downward_consistency}

Assume that all conditions of Theorem~\ref{lm_bic_consistency} hold. Assume that Assumption~\ref{no_gaps} holds. Assume that $\gamma_n \to \infty$ and $\gamma_n = o(n/\sqrt{r_{n,s}})$ for every model $s \in \mathcal{S}$ except for the nonparametric model. Then the selected series form $\hat{s}^{DT}$ is consistent in the following sense: 

\begin{enumerate}

\item

$\hat{s}^{DT} \in \mathcal{S}_u^*$.

\item

$\hat{m}_{n}^{DT}(a_n) \leq m_{n,s'}(a_n)$ for any $s' \in \mathcal{S}_u^*$.

\end{enumerate}

\end{theorem}

\section{Simulations}\label{simulations}

I study the finite sample performance of the LM-BIC and UT model selection procedures in the treatment effects estimation setting. I skip the DT procedure because it requires an additional condition. Suppose that the researcher observes an outcome variable $Y_i$, a treatment dummy variable $D_i$, and two control variables $X_i$ and $Z_i$. The researcher wants to estimate the treatment effects parameter $\beta$ in the model:
\[
Y_i = D_i \beta + g(X_i, Z_i) + \varepsilon_i, \quad E[\varepsilon_i | D_i, X_i, Z_i] = 0
\]

The function $g(X_i, Z_i)$ is unknown, and the researcher wants to obtain a consistent and efficient estimate of $\beta$. She wants to choose among five possible models:

\begin{enumerate}

\item

A fully parametric linear model:
\begin{align}\label{model_p}
Y_i &= D_i \beta + X_i \alpha_1 + Z_i \alpha_2 + \varepsilon_i
\end{align}

\item

A semiparametric partially linear model that is nonparametric in $X_i$:
\begin{align}\label{model_sp_x}
Y_i = D_i \beta + h_1(X_i) + Z_i \alpha_2 + \varepsilon_i
\end{align}

\item

A semiparametric partially linear model that is nonparametric in $Z_i$:
\begin{align}\label{model_sp_z}
Y_i = D_i \beta + X_i \alpha_1 + h_2(Z_i) + \varepsilon_i
\end{align}

\item

A semiparametric additive model that is nonparametric in $X_i$ and $Z_i$ but does not include the interaction term:
\begin{align}\label{model_sp_add}
Y_i = D_i \beta + h_1(X_i)+ h_2(Z_i) + \varepsilon_i
\end{align}

\item

A fully nonparametric model:
\begin{align}\label{modea_np}
Y_i &= D_i \beta + g(X_i,Z_i) + \varepsilon_i
\end{align}

\end{enumerate}

I generate  I consider five data generating processes, each of which corresponds to one of the models above:

\begin{enumerate}

\item

A parametric DGP:
\[
Y_i =  2 D_i + 1 - X_i + 1.5 Z_i + \varepsilon_i
\]

\item

A semiparametric partially linear DGP:
\[
Y_i = 2 D_i + 1 - X_i + 0.25 \exp(X_i - 2) + 1.5 Z_i + \varepsilon_i
\]

\item

Another semiparametric partially linear DGP:
\[
Y_i =  2 D_i + 1 - X_i + 1.5 Z_i + 0.5 \sin(2 (Z_i - 3)) + \varepsilon_i
\]

\item

A semiparametric additive DGP:
\[
Y_i = 2 D_i + 1 - X_i + 0.25 \exp(X_i - 2) + 1.5 Z_i + 0.5 \sin(2 (Z_i - 3)) + \varepsilon_i
\]

\item

A fully nonparametric DGP:
\[
Y_i = 2 D_i + 1 - X_i + 0.25 \exp(X_i - 2) + 1.5 Z_i + 0.5 \sin(2 (Z_i - 3)) + 0.2 X_i Z_i + \sin(X_i Z_i) + \varepsilon_i
\]

\end{enumerate}

In order to use the LM-BIC model selection procedure to choose a model, the researcher has to estimate every model by series methods, replacing all unknown functions with their finite series expansions. In this section, I use the LM-BIC criterion with $\kappa_n = \ln{n}$ and the upward testing procedure with $\gamma_n = 0.05 \sqrt{n}$ to select the best model. I consider two sample sizes, $n = 1,000$ and $n=5,000$. I use power series and use $a_n = 1.5 n^{0.15}$ series terms to approximate each one-dimensional unknown function. For the two-dimensional function, I use $a_n^2$ series terms.

Tables~\ref{model_selection_probs_1}--\ref{beta_properties_5} show the simulated selection probabilities for each of the five data generating processes and the means and mean squared errors of the estimates of $\beta$ in different models. As we can see, the LM-BIC procedure does very well with the sample size $n=5,000$, selecting the ``correct'' model (i.e. the most parsimonious of all correctly specified models) at least 99\% of the time under all five DGPs. The UT procedure does not perform as well, selecting the right model only 83\% of the time under the semiparametric additive DGP. Both procedures have hard time selecting the correct model under the additive DGP with $n=1,000$, selecting the semiparametric additive model less than 25\% of the time. The UT procedure generally selects the right model less frequently than the LM-BIC procedure, but there is one exception. Under the nonparametric DGP with $n=1,000$, the LM-BIC procedure selects the right model 81\% of the time, while the UT procedure always selects it.

Given these simulated model selection probabilities, it is not surprising that post-selection estimates of $\beta$ based on the LM-BIC procedure have lower MSE than the ones based on the UT procedure. Moreover, the MSE of the estimates of $\beta$ based on the LM-BIC is very close to that of the oracle estimator, especially with the sample size $n=5,000$. However, under the nonparametric DGP with $n=1,000$, the post-selection estimate of $\beta$ based on the UT procedure has noticeably lower MSE than the one based on the LM-BIC.

Overall, these results suggest that in most scenarios, the LM-BIC procedure slightly outperforms the UT procedure, but both seem to have good finite sample properties.

\section{Empirical Example}\label{empirical_example}

In this section, I apply the proposed model selection method to the Canadian household gasoline demand dataset from \citet{yatchew_no_2001}. They estimate gasoline demand ($y$), measured as the logarithm of the total distance driven in a given month, as a function of the logarithm of the gasoline price ($PRICE$), the logarithm of the household income ($INCOME$), the logarithm of the age of the primary driver of the car ($AGE$), and other variables ($z$), which include the logarithm of the number of drivers in the household ($DRIVERS$), the logarithm of the household size ($HHSIZE$), an urban dummy, a dummy for singles under 35 years old, and monthly dummies.

\citet{yatchew_no_2001} used several parametric and semiparametric specifications to estimate gasoline demand. The inadequacy of parametric models and importance of semiparametric ones in gasoline demand estimation was first pointed out by \citet{hausman_newey_1995} and \citet{schmalensee_stoker_1999}. \citet{yatchew_no_2001} follow these papers in using semiparametric specifications for gasoline demand. In \citet{korolev_2018}, I apply the semiparametric LM type specification test to one of their models, which is nonparametric in $AGE$ but is parametric in all other variables, and find no evidence against it. 

In this paper, I use a different approach. I consider seven candidate models:

\begin{enumerate}

\item

A fully parametric model (``P''):
\begin{align}\label{gasoline_p}
y = \alpha_1 PRICE + \alpha_2 INCOME + \gamma_0 + \gamma_1 AGE + z' \beta + \varepsilon
\end{align}

\item

A semiparametric model which is nonparametric in $AGE$ but parametric in all other variables (``SP-$AGE$''):
\begin{align}\label{gasoline_sp_age}
y = \alpha_1 PRICE + \alpha_2 INCOME + g(AGE) + z' \beta + \varepsilon
\end{align}

\item

A semiparametric model which is nonparametric in $PRICE$ but parametric in all other variables (``SP-$PRICE$''):
\begin{align}\label{gasoline_sp_price}
y = \gamma_1 AGE + \alpha_2 INCOME + g(PRICE) + z' \beta + \varepsilon
\end{align}

\item

A semiparametric model which is nonparametric in $INCOME$ but parametric in all other variables (``SP-$INCOME$''):
\begin{align}\label{gasoline_sp_income}
y = \gamma_1 AGE + \alpha_1 PRICE + g(INCOME) + z' \beta + \varepsilon
\end{align}

\item

A semiparametric additive model which is nonparametric in $AGE$, $PRICE$, and $INCOME$ separately, but does not include any interaction terms between them (``SP-Additive''):
\begin{align}\label{gasoline_sp_add3}
y = g_1(AGE) + g_2(PRICE) + g_3(INCOME) + z' \beta + \varepsilon
\end{align}

\item

A semiparametric model which is nonparametric in $AGE$ and $PRICE$ but parametric in all other variables (``SP-$AGE \& PRICE$''):
\begin{align}\label{gasoline_sp_age_price}
y = \alpha_2 INCOME + g(AGE, PRICE) + z' \beta + \varepsilon
\end{align}

\item

A nonparametric (in all continuous variables) model (``NP''):
\begin{align}\label{gasoline_np}
y = h(PRICE, AGE, INCOME, z_1) + z_2' \lambda + \varepsilon
\end{align}

\end{enumerate}

I estimate each of these seven models using series methods, replacing all unknown functions with their finite series expansions, and compute the value of the model selection criterion as in Equation~\ref{msc_criterion}, both for the homoskedastic and heteroskedastic versions of the LM type test statistic. Following \citet{korolev_2018}, to construct the regressors $P^{k_n}$ used to evaluate the test statistic, I use $a_n = 3$ power series terms in $AGE$, $PRICE$, and $INCOME$, $j_n = 2$ power series terms in $DRIVERS$ and $HHSIZE$, and the set of dummies discussed above. I then use pairwise interactions (tensor products) of univariate power series, and add all possible three, four, and five element interactions between $AGE$, $PRICE$, $INCOME$, $DRIVERS$, and $HHSIZE$, without using higher powers in these interaction terms to avoid multicollinearity. Table~\ref{empirical_example_results} presents the results.

As we can see, both the homoskedastic and heteroskedastic versions of the LM-BIC model selection criterion choose the sepimarametric model~\ref{gasoline_sp_age}, which is nonparametric in $AGE$ but parametric in all other variables. This is in line with the conclusions of \citet{yatchew_no_2001}, who conclude that it is $AGE$ that nonlinearly affects gasoline demand, and of \citet{korolev_2018}, who finds no evidence against model~\ref{gasoline_sp_age} when tested against model~\ref{gasoline_np}.

However, if I use the UT or DT procedure with $\gamma_n = 0.05 \sqrt{n} = 3.95$ assuming that the errors are homoskedastic, then I end up selecting the fully parametric model. In contrast, if I assume that the errors are heteroskedastic, then the parametric model is rejected, and the semiparametric model~\ref{gasoline_sp_age} is selected. 

This example illustrates a couple of potential issues with the proposed model selection procedures. First, even though asymptotically all procedures are consistent, in finite samples they may select different models. More specifically, if one assumes that the errors are homoskedastic, the parametric model ends up being selected by the UT and DT procedures, despite being rejected by the LM type specification test and not being selected by the LM-BIC procedure.

Second, the critical value $\gamma_n$ is somewhat arbitrary, and the UT and DT procedures may be very sensitive to the choice of critical values. If I chose a slightly higher $\gamma_n$, the parametric model would be selected by the UT and DT procedures even in the heteroskedastic case. If I chose a slightly lower value of $\gamma_n$, the parametric model would not be selected by these procedures even in the homoskedastic case. 

While these two issues are important, my model selection procedure is not the only one affected by them. In fact, many existing model selection procedures are heavily dependent on tuning parameters choice, and developing ways to choose tuning parameters optimally in model selection may be an important avenue for future research.

\section{Conclusion}\label{conclusion}

In this paper, I develop a new model selection procedure for parametric, semiparametric, and nonparametric models. It combines the semiparametric LM type test from \citet{korolev_2018} and the BIC model selection procedure. The LM-BIC procedure aims to select the model with the lowest value of the semiparametric LM type test statistic but additionally penalizes complicated models. I prove that the resulting procedure is consistent, in the sense that it selects the lowest dimensional correctly specified model. I also develop an upward testing and downward testing procedures based on the semiparametric LM type specification test and prove their consistency. The proposed LM-BIC and UT procedures perform well in simulations.

In future work, it would be interesting to study how to select the number of terms in series expansions that are used to estimate different models and how to optimally choose $\kappa_n$, the penalty parameter for the LM-BIC, and $\gamma_n$, the critical value for the UT and DT procedures. It would also be important to develop methods for inference on post-selection estimates.

\clearpage

\renewcommand\thesection{\Alph{section}}

\renewcommand{\thetheorem}{A.\arabic{theorem}}

\renewcommand{\thelemma}{A.\arabic{lemma}}

\renewcommand{\theassumption}{A.\arabic{assumption}}

\renewcommand{\theremark}{A.\arabic{remark}}

\setcounter{theorem}{0}

\setcounter{lemma}{0}

\setcounter{assumption}{0}

\setcounter{remark}{0}

\begin{appendices}

\section{Tables and Figures}\label{appendix_tables_figures}

\begin{table}[h]
\begin{center}
\caption{Model Selection Probabilities, DGP 1}\label{model_selection_probs_1}
\begin{tabular}{l | c c c c c}
Model & P & SP-X & SP-Z & SP-ADD & NP \\ \hline \hline
& \multicolumn{5}{c}{LM-BIC} \\
Selection probability, $n=1,000$ & 0.9860 & 0.0055 & 0.0085 & 0.0000 & 0.0000 \\
Selection probability, $n=5,000$ & 0.9980 & 0.0015 & 0.0005 & 0.0000 & 0.0000 \\ \\ \hline\hline
& \multicolumn{5}{c}{UT Procedure} \\
Selection probability, $n=1,000$ & 0.9380	& 0.0145	& 0.0105	& 0.0020	& 0.0350 \\
Selection probability, $n=5,000$ &  0.9965 &	0.0010 & 	0.0010 &	0.0000 &	0.0015\\
\end{tabular}
\end{center}
\footnotesize{The table shows the simulated probabilities of selecting each of the five models under consideration under DGP 1, which corresponds to the parametric model $Y_i = 1 + 2 D_i - X_i + 1.5 Z_i + \varepsilon_i$.
Model P refers to the fully parametric model~\ref{model_p}. Model SP-X corresponds to the semiparametric model ~\ref{model_sp_x}. Model SP-Z corresponds to the semiparametric model ~\ref{model_sp_z}. Model SP-ADD corresponds to the semiparametric model ~\ref{model_sp_add}. Model NP corresponds to the fully nonparametric model ~\ref{modea_np}. Results are based on $B=2,000$ simulation draws.}
\end{table}

\begin{table}[h]
\begin{center}
\caption{Mean and MSE of $\hat{\beta}$, DGP 1}\label{beta_properties_1}
\begin{tabular}{l | c c c c c cc }
Model & P & SP-X & SP-Z & SP-ADD & NP & SM & UT \\ \hline \hline
 &  \multicolumn{7}{c}{$n=1,000$} \\
Mean of $\hat{\beta}$ & 2.003 & 2.003 & 2.004 & 2.004 & 2.007 & 2.002 & 2.002 \\
MSE of $\hat{\beta}$ & 0.0221 & 0.0255 & 0.0261 & 0.0278 & 0.0370 & 0.0225 & 0.0233 \\
\\ \hline \hline
 &  \multicolumn{7}{c}{$n=5,000$} \\
Mean of $\hat{\beta}$ & 2.000 & 2.001 & 2.001 & 2.001 & 2.001 & 2.000 & 2.000 \\
MSE of $\hat{\beta}$ & 0.0046 & 0.0055 & 0.0054 & 0.0061 & 0.0076 & 0.0046 & 0.0046
\end{tabular}
\end{center}
\footnotesize{The table shows the simulated means and mean squared errors of the estimates of $\beta$ for the five models under consideration under DGP 1, which corresponds to the model $Y_i = 1 + 2 D_i - X_i + 1.5 Z_i + \varepsilon_i$. The true value is $\beta = 2$. Model P refers to the fully parametric model~\ref{model_p}. Model SP-X corresponds to the semiparametric model ~\ref{model_sp_x}. Model SP-Z corresponds to the semiparametric model ~\ref{model_sp_z}. Model SP-ADD corresponds to the semiparametric model ~\ref{model_sp_add}. Model NP corresponds to the fully nonparametric model ~\ref{modea_np}. Model SM corresponds to the model selected by the LM-BIC procedure.  Model UT corresponds to the model selected by the UT procedure. Results are based on $B=2,000$ simulation draws.}
\vspace{-2cm}\end{table}

\clearpage

\begin{table}[h]
\begin{center}
\caption{Model Selection Probabilities, DGP 2}\label{model_selection_probs_2}
\begin{tabular}{l | c c c c c}
Model & P & SP-X & SP-Z & SP-ADD & NP \\ \hline \hline
& \multicolumn{5}{c}{LM-BIC} \\
Selection probability, $n=1,000$ & 0.2890 & 0.6830 & 0.0265 & 0.0015 & 0.0000 \\
Selection probability, $n=5,000$ & 0.0000	& 1.0000 & 0.0000 & 0.0000 & 0.0000 \\ \\ \hline\hline
& \multicolumn{5}{c}{UT Procedure} \\
Selection probability, $n=1,000$ & 0.3905	& 0.5180	& 0.0235	& 0.0130	& 0.0550 \\
Selection probability, $n=5,000$ & 0.0100	& 0.9875	& 0.0000	& 0.0020	& 0.0005 \\
\end{tabular}
\end{center}
\footnotesize{The table shows the simulated probabilities of selecting each of the five models under consideration under DGP 2, which corresponds to the model $Y_i = 1 + 2 D_i - X_i + 0.25 \exp(X_i - 2) + 1.5 Z_i + \varepsilon_i$. Model P refers to the fully parametric model~\ref{model_p}. Model SP-X corresponds to the semiparametric model ~\ref{model_sp_x}. Model SP-Z corresponds to the semiparametric model ~\ref{model_sp_z}. Model SP-ADD corresponds to the semiparametric model ~\ref{model_sp_add}. Model NP corresponds to the fully nonparametric model ~\ref{modea_np}. Results are based on $B=2,000$ simulation draws.}
\end{table}

\begin{table}[h]
\begin{center}
\caption{Mean and MSE of $\hat{\beta}$, DGP 2}\label{beta_properties_2}
\begin{tabular}{l | c c c c c c c }
Model & P & SP-X & SP-Z & SP-ADD & NP & SM & UT \\ \hline \hline
 &  \multicolumn{7}{c}{$n=1,000$} \\
Mean of $\hat{\beta}$ & 1.948 & 2.003 & 1.962 & 2.003 & 2.004 & 1.992 & 1.986 \\
MSE of $\hat{\beta}$ & 0.0263 & 0.0274 & 0.0291 & 0.0300 & 0.0378 & 0.0273 & 0.0285 \\
\\ \hline \hline
 &  \multicolumn{7}{c}{$n=5,000$} \\
Mean of $\hat{\beta}$ & 1.945 & 1.999 & 1.959 & 1.999 & 1.998 & 1.999 & 1.998 \\
MSE of $\hat{\beta}$ & 0.0078 & 0.0054 & 0.0070 & 0.0058 & 0.0072 & 0.0054 & 0.0054
\end{tabular}
\end{center}
\footnotesize{The table shows the simulated means and mean squared errors of the estimates of $\beta$ for the five models under consideration under DGP 2, which corresponds to the model $Y_i = 1 + 2 D_i - X_i + 0.25 \exp(X_i - 2) + 1.5 Z_i + \varepsilon_i$. The true value is $\beta = 2$. Model P refers to the fully parametric model~\ref{model_p}. Model SP-X corresponds to the semiparametric model ~\ref{model_sp_x}. Model SP-Z corresponds to the semiparametric model ~\ref{model_sp_z}. Model SP-ADD corresponds to the semiparametric model ~\ref{model_sp_add}. Model NP corresponds to the fully nonparametric model ~\ref{modea_np}. Model SM corresponds to the model selected by the LM-BIC procedure. Model UT corresponds to the model selected by the UT procedure. Results are based on $B=2,000$ simulation draws.}
\vspace{-2cm}\end{table}

\clearpage

\begin{table}[h]
\begin{center}
\caption{Model Selection Probabilities, DGP 3}\label{model_selection_probs_3}
\begin{tabular}{l | c c c c c}
Model & P & SP-X & SP-Z & SP-ADD & NP \\ \hline \hline
& \multicolumn{5}{c}{LM-BIC} \\
Selection probability, $n=1,000$ & 0.4685	& 0.0065 & 0.5235 & 0.0015 & 0.0000 \\
Selection probability, $n=5,000$ & 0.0010 & 0.0000 & 0.9985 & 0.0005 & 0.0000 \\ \\ \hline\hline
& \multicolumn{5}{c}{UT Procedure} \\
Selection probability, $n=1,000$ &  0.5000 & 0.0100	& 0.4110	& 0.0160	& 0.0630 \\
Selection probability, $n=5,000$ & 0.0520	& 0.0000 & 0.9465 &	 0.0005 &	0.0010
\end{tabular}
\end{center}
\footnotesize{The table shows the simulated probabilities of selecting each of the five models under consideration under DGP 3, which corresponds to the model $Y_i = 1 + 2 D_i - X_i + 1.5 Z_i + 0.5 \sin(2 (Z_i - 3)) + \varepsilon_i$. Model P refers to the fully parametric model~\ref{model_p}. Model SP-X corresponds to the semiparametric model ~\ref{model_sp_x}. Model SP-Z corresponds to the semiparametric model ~\ref{model_sp_z}. Model SP-ADD corresponds to the semiparametric model ~\ref{model_sp_add}. Model NP corresponds to the fully nonparametric model ~\ref{modea_np}. Results are based on $B=2,000$ simulation draws.}
\end{table}

\begin{table}[h]
\begin{center}
\caption{Mean and MSE of $\hat{\beta}$, DGP 3}\label{beta_properties_3}
\begin{tabular}{l | c c c c c c c}
Model & P & SP-X & SP-Z & SP-ADD & NP & SM & UT \\ \hline \hline
 &  \multicolumn{7}{c}{$n=1,000$} \\
Mean of $\hat{\beta}$ & 2.203 & 2.152 & 1.997 & 2.001 & 1.999 & 2.071 & 2.083 \\
MSE of $\hat{\beta}$ & 0.0648 & 0.0506 & 0.0272 & 0.0298 & 0.0385 & 0.0448 & 0.0484 \\
\\ \hline \hline
 &  \multicolumn{7}{c}{$n=5,000$} \\
Mean of $\hat{\beta}$ & 2.206 & 2.158 & 2.001 & 2.000 & 2.003 & 2.001 & 2.009 \\
MSE of $\hat{\beta}$ & 0.0466 & 0.0301 & 0.0047 & 0.0055 & 0.0072 & 0.0048 & 0.0064
\end{tabular}
\end{center}
\footnotesize{The table shows the simulated means and mean squared errors of the estimates of $\beta$ for the five models under consideration under DGP 3, which corresponds to the model $Y_i = 1 + 2 D_i - X_i + 1.5 Z_i + 0.5 \sin(2 (Z_i - 3)) + \varepsilon_i$. The true value is $\beta = 2$. Model P refers to the fully parametric model~\ref{model_p}. Model SP-X corresponds to the semiparametric model ~\ref{model_sp_x}. Model SP-Z corresponds to the semiparametric model ~\ref{model_sp_z}. Model SP-ADD corresponds to the semiparametric model ~\ref{model_sp_add}. Model NP corresponds to the fully nonparametric model ~\ref{modea_np}. Model SM corresponds to the model selected by the LM-BIC procedure. Model UT corresponds to the model selected by the UT procedure. Results are based on $B=2,000$ simulation draws.}
\vspace{-2cm}\end{table}

\clearpage

\begin{table}[h]
\begin{center}
\caption{Model Selection Probabilities, DGP 4}\label{model_selection_probs_4}
\begin{tabular}{l | c c c c c}
Model & P & SP-X & SP-Z & SP-ADD & NP \\ \hline \hline
& \multicolumn{5}{c}{LM-BIC} \\
Selection probability, $n=1,000$ & 0.1930 & 0.2895 & 0.2855 & 0.2320 & 0.0000 \\
Selection probability, $n=5,000$ & 0.0000 & 0.0025 & 0.0015 & 0.9960 & 0.0000 \\ \\ \hline\hline
& \multicolumn{5}{c}{UT Procedure} \\
Selection probability, $n=1,000$ & 0.1450	& 0.2830	& 0.2795	& 0.2305	& 0.0620 \\
Selection probability, $n=5,000$ & 0.0000	& 0.0815	& 0.0870	& 0.8300	& 0.0015 \\
\end{tabular}
\end{center}
\footnotesize{The table shows the simulated probabilities of selecting each of the five models under consideration under DGP 4, which corresponds to the model $Y_i = 1 + 2 D_i - X_i + 0.25 \exp(X_i - 2) + 1.5 Z_i + 0.5 \sin(2 (Z_i - 3)) + \varepsilon_i$. Model P refers to the fully parametric model~\ref{model_p}. Model SP-X corresponds to the semiparametric model ~\ref{model_sp_x}. Model SP-Z corresponds to the semiparametric model ~\ref{model_sp_z}. Model SP-ADD corresponds to the semiparametric model ~\ref{model_sp_add}. Model NP corresponds to the fully nonparametric model ~\ref{modea_np}. Results are based on $B=2,000$ simulation draws.}
\end{table}

\begin{table}[h]
\begin{center}
\caption{Mean and MSE of $\hat{\beta}$, DGP 4}\label{beta_properties_4}
\begin{tabular}{l | c c c c c c c}
Model & P & SP-X & SP-Z & SP-ADD & NP & SM & UT \\ \hline \hline
 &  \multicolumn{7}{c}{$n=1,000$} \\
Mean of $\hat{\beta}$ & 2.149 & 2.154 & 1.960 & 2.005 & 2.002 & 2.048 &  2.045 \\
MSE of $\hat{\beta}$ & 0.0452 & 0.0506 & 0.0280 & 0.0291 & 0.0370 & 0.0399 & 0.0404 \\
\\ \hline \hline
 &  \multicolumn{7}{c}{$n=5,000$} \\
Mean of $\hat{\beta}$ & 2.149 & 2.156 & 1.959 & 1.998 & 1.998 & 1.999 & 2.006 \\
MSE of $\hat{\beta}$ & 0.0271 & 0.0294 & 0.0071 & 0.0058 & 0.0075 & 0.0058 & 0.0079
\end{tabular}
\end{center}
\footnotesize{The table shows the simulated means and mean squared errors of the estimates of $\beta$ for the five models under consideration under DGP 4, which corresponds to the model $Y_i = 1 + 2 D_i - X_i + 0.25 \exp(X_i - 2) + 1.5 Z_i + 0.5 \sin(2 (Z_i - 3)) + \varepsilon_i$. The true value is $\beta = 2$. Model P refers to the fully parametric model~\ref{model_p}. Model SP-X corresponds to the semiparametric model ~\ref{model_sp_x}. Model SP-Z corresponds to the semiparametric model ~\ref{model_sp_z}. Model SP-ADD corresponds to the semiparametric model ~\ref{model_sp_add}. Model NP corresponds to the fully nonparametric model ~\ref{modea_np}. Model SM corresponds to the model selected by the LM-BIC procedure. Model UT corresponds to the model selected by the UT procedure. Results are based on $B=2,000$ simulation draws.}
\vspace{-2cm}\end{table}

\clearpage

\begin{table}[h]
\begin{center}
\caption{Model Selection Probabilities, DGP 5}\label{model_selection_probs_5}
\begin{tabular}{l | c c c c c}
Model & P & SP-X & SP-Z & SP-ADD & NP \\ \hline \hline
& \multicolumn{5}{c}{LM-BIC} \\
Selection probability, $n=1,000$ & 0.0200	& 0.0020	& 0.1640	& 0.0030	& 0.8110 \\
Selection probability, $n=5,000$ & 0.0000	& 0.0000	& 0.0000	& 0.0000	& 1.0000 \\ \\ \hline\hline
& \multicolumn{5}{c}{UT Procedure} \\
Selection probability, $n=1,000$ & 0.0000	& 0.0000	& 0.0000	& 0.0000	& 1.0000  \\
Selection probability, $n=5,000$ & 0.0000	& 0.0000	& 0.0000	& 0.0000	& 1.0000 
\end{tabular}
\end{center}
\footnotesize{The table shows the simulated probabilities of selecting each of the five models under consideration under DGP 5, which corresponds to the model $Y_i = 1 + 2 D_i - X_i + 0.25 \exp(X_i - 2) + 1.5 Z_i + 0.5 \sin(2 (Z_i - 3)) + 0.2 X_i Z_i + \sin(X_i Z_i) + \varepsilon_i$. Model P refers to the fully parametric model~\ref{model_p}. Model SP-X corresponds to the semiparametric model ~\ref{model_sp_x}. Model SP-Z corresponds to the semiparametric model ~\ref{model_sp_z}. Model SP-ADD corresponds to the semiparametric model ~\ref{model_sp_add}. Model NP corresponds to the fully nonparametric model ~\ref{modea_np}. Results are based on $B=2,000$ simulation draws.}
\end{table}

\begin{table}[h]
\begin{center}
\caption{Mean and MSE of $\hat{\beta}$, DGP 5}\label{beta_properties_5}
\begin{tabular}{l | c c c c c c c}
Model & P & SP-X & SP-Z & SP-ADD & NP & SM  & UT\\ \hline \hline
 &  \multicolumn{7}{c}{$n=1,000$} \\
Mean of $\hat{\beta}$ & 1.612	& 1.511	& 1.227	& 1.234	& 1.879 & 1.773 & 1.879 \\
MSE of $\hat{\beta}$ & 0.1795	& 0.2727	& 0.6301	& 0.6225	& 0.0543 & 0.1558 & 0.0543 \\
\\ \hline \hline
 &  \multicolumn{7}{c}{$n=5,000$} \\
Mean of $\hat{\beta}$ &  0.619	& 1.546	& 1.232	& 1.315	& 1.993 & 1.993 & 1.993 \\
MSE of $\hat{\beta}$ & 0.1508	& 0.2127	& 0.5965	& 0.4766	& 0.0077 & 0.0077 & 0.0077
\end{tabular}
\end{center}
\footnotesize{The table shows the simulated means and mean squared errors of the estimates of $\beta$ for the five models under consideration under DGP 5, which corresponds to the model $Y_i = 1 + 2 D_i - X_i + 0.25 \exp(X_i - 2) + 1.5 Z_i + 0.5 \sin(2 (Z_i - 3)) + 0.2 X_i Z_i + \sin(X_i Z_i) + \varepsilon_i$. The true value is $\beta = 2$. Model P refers to the fully parametric model~\ref{model_p}. Model SP-X corresponds to the semiparametric model ~\ref{model_sp_x}. Model SP-Z corresponds to the semiparametric model ~\ref{model_sp_z}. Model SP-ADD corresponds to the semiparametric model ~\ref{model_sp_add}. Model NP corresponds to the fully nonparametric model ~\ref{modea_np}. Model SM corresponds to the model selected by the LM-BIC procedure. Model UT corresponds to the model selected by the UT procedure. Results are based on $B=2,000$ simulation draws.}
\vspace{-2cm}\end{table}

\clearpage

\begin{table}[h]
\begin{center}
\caption{Model Selection in Gasoline Demand Estimation}\label{empirical_example_results}
\begin{tabular}{l | c | c c | c c }
& & \multicolumn{2}{c}{Homoskedastic Version} & \multicolumn{2}{c}{Heteroskedastic Version} \\
Model & $r$ & $t$-stat & $MSC$ & $t$-stat HC & $MSC$ HC \\ \hline\hline
P & 91 & 3.891 & -48.3 & 4.065 & -48.1 \\
SP-$AGE$ & 89 & 0.227 & \textbf{-51.4} & 0.526 & \textbf{-51.1} \\
SP-$PRICE$ & 89 & 3.983 & -47.6 & 4.058 & -47.6 \\
SP-$INCOME$ & 89 & 3.890 & -47.7 & 4.143 & -47.5 \\
SP-Additive & 85 & 0.258 & -50.2 & 0.531 & -49.9 \\
SP-$AGE \& PRICE$ & 78 & 0.647 & -47.7 & 0.700 & -47.6 \\
NP & 0 & 0 & 0 & 0 & 0
\end{tabular}
\end{center}
\footnotesize{The table shows the values of the LM type test statistic $t$ (given in equation 3.7 in \citet{korolev_2018}) and its heteroskedasticity robust version (given in equation 3.9 in \citet{korolev_2018}), as well as the corresponding values of the model selection criterion from equation~\ref{msc_criterion}, for the models from Section~\ref{empirical_example}. Model ``P'' corresponds to equation~\ref{gasoline_p}. Model ``SP-$AGE$'' corresponds to equation~\ref{gasoline_sp_age}. Model ``SP-$PRICE$'' corresponds to equation~\ref{gasoline_sp_price}. Model ``SP-$INCOME$'' corresponds to equation~\ref{gasoline_sp_income}. Model ``SP-Additive'' corresponds to equation~\ref{gasoline_sp_add3}. Model ``SP-$AGE \& PRICE$'' corresponds to equation~\ref{gasoline_sp_age_price}. Model ``NP'' corresponds to equation~\ref{gasoline_np}. The minimum values of the model selection criterion are shown in bold and correspond to the selected model. The column $r$ shows the number of restrictions associated with different models. The sample size is $n=6,230$.}
\end{table}

\clearpage

\section{Proofs}\label{appendix_proofs}

\begin{proof}[Proof of Theorem~\ref{lm_bic_consistency}.]

First, suppose that $s \in \mathcal{S}_u^*$. Then, as shown in Theorem 1 in \citet{korolev_2018}, as $n \to \infty$,
\[
\frac{\xi_s - r_{n,s}}{\sqrt{2 r_{n,s}}} = O_p(1)
\]

Thus,
\[
MSC(s) = \frac{\xi - r_{n,s}}{\sqrt{2 r_{n,s}}} - \sqrt{\frac{r_{n,s}}{2}} (\kappa_n - 1) \to -\infty,
\]
because $\sqrt{\frac{r_{n,s}}{2}} (\kappa_n - 1) \to \infty$.

Next, suppose that $s \notin \mathcal{S}_u^*$ (i.e. $s \in \mathcal{S}^*(c)$ and $c \notin \mathcal{C}^*$, or $s \notin \mathcal{S}^*(c)$ for any $c \in \mathcal{C}$). Then the series form $s$ has a misspecification error that does not vanish even asymptotically. Therefore, as shown in Theorem 4 in \citet{korolev_2018},
\[
\frac{\sqrt{r_{n,s}}}{n} \frac{\xi - r_{n,s}}{\sqrt{2 r_{n,s}}} \overset{p}{\to} \Delta/\sqrt{2},
\]

Thus,
\[
MSC(s) \frac{\sqrt{2 r_{n,s}}}{n} = \sqrt{2} \frac{\sqrt{r_{n,s}}}{n} \frac{\xi - r_{n,s}}{\sqrt{2 r_{n,s}}} - \frac{r_{n,s}(\kappa_n - 1)}{n} \overset{p}{\to} \Delta,
\]
because $\kappa_n = o(n/r_{n,s})$, and thus $MSC(s) \to +\infty$.

Thus, the series form $\hat{s}$ chosen by the model selection procedure satisfies $s \in \mathcal{S}_u^*$ w.p.a. 1. 

Next, suppose that there are two series forms $s_1 \in \mathcal{S}_u^*$ and $s_2 \in \mathcal{S}_u^*$, but $m_{n,s_1}(a_n) < m_{n,s_2}(a_n)$ and hence $r_{n,s_1}(a_n) > r_{n,s_2}(a_n)$ (i.e. the series form $s_1$ is more restrictive). Then
\[
MSC(s_1) - MSC(s_2) = \frac{\xi - r_{n,s_1}}{\sqrt{2 r_{n,s_1}}} - \frac{\xi - r_{n,s_2}}{\sqrt{2 r_{n,s_2}}} - \sqrt{\frac{r_{n,s_1} - r_{n,s_2}}{2}} (\kappa_n - 1) \to -\infty
\]

Thus, asymptotically $s_1$ will be preferred to $s_2$. Therefore, $m_{n,\hat{s}}(a_n) \leq m_{n,s'}(a_n)$ w.p.a. 1 for any $s' \in \mathcal{S}_u^*$.

\end{proof}

\begin{proof}[Proof of Theorem~\ref{lm_bic_consistency_hc}.]

Suppose that $s \in \mathcal{S}_u^*$. Then, as shown in Theorem 3 in \citet{korolev_2018}, as $n \to \infty$,
\[
\frac{\xi_{HC,s} - r_{n,s}}{\sqrt{2 r_{n,s}}} = O_p(1)
\]

The rest of the proof is similar to the proof of Theorem~\ref{lm_bic_consistency}.

\end{proof}

\begin{proof}[Proof of Theorem~\ref{upward_consistency}.]

Fix $a_n$. First, suppose that $s \notin \mathcal{S}_u^*$. Then
\[
\frac{\sqrt{r_{n,s}}}{n} t_{s, m_{n,s}} = \frac{\sqrt{r_{n,s}}}{n} \frac{\xi - r_{n,s}}{\sqrt{2 r_{n,s}}} \overset{p}{\to} \Delta/\sqrt{2},
\]
and
\[
t_{s, m_{n,s}}/\gamma_n \overset{p}{\to} \infty
\]

Thus, if for $m^*$ there is no series form $s$ with $m_{n,s} = m^*$ such that $s \in \mathcal{S}_u^*$, then all models with $m_{n,s} = m^*$ will be rejected w.p.a. 1. Thus, $\hat{s}^{UT} \in \mathcal{S}_u^*$ w.p.a. 1.

Now suppose that $s \in \mathcal{S}_u^*$. Then
\[
t_{s, m_{n,s}} = \frac{\xi_s - r_{n,s}}{\sqrt{2 r_{n,s}}} = O_p(1),
\]
and
\[
t_{s, m_{n,s}} < \gamma_n \text{ w.p.a. 1}
\]

Hence, if for $m^*$ there is a series form $s$ with $m_{n,s} = m^*$ such that $s \in \mathcal{S}_u^*$, then $\hat{m}_{n}^{UT} \leq m^*$. Therefore, $\hat{m}_{n}^{UT}(a_n) \leq m_{n,s'}(a_n)$ for all $s' \in \mathcal{S}_u^*$ w.p.a. 1.

\end{proof}

\begin{proof}[Proof of Theorem~\ref{downward_consistency}.]

First, note that by the construction of $\hat{s}^{DT}$, $t_{\hat{s}^{DT}} \leq \gamma_n$. Thus, by the same logic as in the proof of Theorem~\ref{upward_consistency}, $\hat{s}^{DT} \in \mathcal{S}_u^*$ w.p.a. 1. Moreover, under Assumption~\ref{no_gaps}, $\mathcal{M}_0^{**}$ w.p.a. 1 includes all $m^{**}$ for which there is a series form $s$ with $m_{n,s} = m^{**}$ such that $s \in \mathcal{S}_u^*$. Because $m^{DT} = \min_{m \in \mathcal{M}_0^{**}}{m}$, $\hat{m}_{n}^{DT}(a_n) \leq m_{n,s'}(a_n)$ for all $s' \in \mathcal{S}_u^*$ w.p.a. 1.

\end{proof}

\end{appendices}

\clearpage

\bibliography{literature_model_selection}

\end{document}